\documentclass[runningheads,fleqn]{llncs}
\usepackage{amssymb}
\usepackage{amsmath}
\usepackage{latexsym}
\usepackage{amsfonts}%
\usepackage{algorithm}
\usepackage{algorithmic}

\begin{document}
\mainmatter

\title{A note on Thue games}

\titlerunning{A note on Thue games} 

\author{Robert Merca\c{s}\thanks{Work supported by the P.R.I.M.E. programme of  DAAD co-funded by BMBF and EU's 7th Framework Programme (grant 605728).
} \and Dirk Nowotka}

\authorrunning{Merca\c{s} \and Nowotka}

\institute{Department of Computer Science, Kiel University, Germany,
\email{RobertMercas@gmail.com, dn@informatik.uni-kiel.de}
}

\maketitle

\begin{abstract}
In this work we improve on a result from~\cite{GryKosZma15}. In particular, we investigate the situation where a word is constructed jointly by two players who alternately append letters to the end of an existing word. One of the players (Ann) tries to avoid (non-trivial) repetitions, while the other one (Ben) tries to enforce them. We show a construction that is closer to the lower bound showed in~\cite{GryKozMic13} using entropy compression, and building on the probabilistic arguments based on a version of the Lov\'asz Local Lemma from~\cite{Peg11}. We provide an explicit strategy for Ann to avoid (non-trivial) repetitions over a $7$-letter alphabet.
\end{abstract}

\section{Results}
In~\cite{GryKosZma15} the authors give an explicit strategy for Ann to win the game by considering two cases separately. Namely the avoidability of non-trivial squares of words with odd length and the avoidability of squares of words with even length. Since both these strategies imply the use of a ternary alphabet, the result is just a cartesian product of the two disjoint alphabets.

For the avoidability of squares of even length, the authors propose that Ann uses the letters of the well known Thue-Morse word~\cite{Thu12}, in the order of their appearance. We commit to the same strategy as it seems the most viable. 

For the other case, they show that if Ann choses a favourite letter which she continues to repeat as long as Ben does not repeat it, this avoids non-trivial squares of words of odd length, as long as Ann changes her choice of a favourite letter as soon as Ben repeats it, to the letter that Ben has not used in his previous two choices. We will tweak a bit this part of the strategy.

Our strategy follows the same lines as that from~\cite{GryKosZma15}, but with a refined counting. In particular, we now already consider the pairs of letters, and not only a partial alphabet for each of the situations. Therefore, we consider pairs of letters from $\{\textbf 0,\textbf 1\} \times \{\textbf a,\textbf b,\textbf c\}$, where the second positions in the pair will render for Ann the Thue-Morse word. Furthermore, we will consider an extra letter $(\textbf 2,\textbf d)$ which has a unique first and second component among the others.
Ann will still chose a favourite letter, but only for the first position of the pair, while she will adhere to the previous restrictions imposed by the Thue-Morse words, for the second component. She will repeat the letter for the first component as long as Ben has not made the same choice for the first component, previously.

First observe that according to the previous explanations, our considered alphabet is formed of pairs of letters, namely:

$$\Sigma = \{(\textbf 0,\textbf a), (\textbf 0,\textbf b), (\textbf 0,\textbf c), (\textbf 1,\textbf a), (\textbf 1,\textbf b), (\textbf 1,\textbf c), (\textbf 2,\textbf d)\}.$$

The following algorithm describes the game based on Ann's  strategy of choosing a favourite, when she starts the game (the other situation is similar):

\begin{algorithm}[h!]
\caption{Construction of the word using Ann's strategy of favourite letter}
\begin{algorithmic}[1]\label{minper}
\STATE Let $\tau$ be the Thue-Morse word.
\STATE $Ann_{\tt F}=(\textbf 0,\textbf a)$.		\hfill {\scriptsize \tt//$Ann_{\tt F}$ is the current favourite letter of Ann}
\STATE $Ben_{\tt  M}={\bf read.in()}$.\hfill {\scriptsize \tt//$Ben_{\tt M}$ is the current choice of Ben}
\STATE {\bf Append}$(Ann_{\tt F})$. \hfill {\scriptsize \tt//we add the first character of the word}
\STATE $Ann_{\tt F}=(1-\lceil \frac{Ben_{\tt M}}{2}\rceil,\textbf b)$.	\hfill {\scriptsize \tt//next $Ann_{\tt F}$ must be different from current $Ben_{\tt M}$}
\STATE {\bf Append}$(Ben_{\tt M})$, {\bf Append}$(Ann_{\tt F})$. \hfill {\scriptsize \tt//we add the current characters}
\STATE $Ben_{\tt P}=Ben_{\tt M}$. \hfill {\scriptsize \tt//$Ben_{\tt P}$ is the previous move of Ben}
\STATE $Ben_{\tt M}={\bf read.in()}$.
\STATE $count=3$.		\hfill {\scriptsize \tt//recall the next character of $\tau$ to be used}.
\WHILE {Game Played}
	\IF {$Ann_{\tt F}[1]==\textbf 2$}
		\STATE $Ann_{\tt F}=(1-\lfloor \frac{Ben_{\tt M}[1]}{2}\rfloor Ben_{\tt P}[1] - Ben_{\tt M}[1] (2-Ben_{\tt M}[1]),\tau[count])$.
	\ELSIF {$Ben_{\tt M}==Ann_{\tt F}$}
		\STATE $Ann_{\tt F}=(3-Ben_{\tt P}-Ben_{\tt M},\lfloor \frac{Ben_{\tt P}[2]}{\textbf d}\rfloor \textbf d +(1-\lfloor \frac{Ben_{\tt P}[2]}{\textbf d}\rfloor)\tau[count])$.
	\ELSE 
	\STATE {$Ann_{\tt F}[2]=\tau[count]$} \hfill{\scriptsize \tt{/}{/}update the second component of Ann's favourite}
	\ENDIF
\STATE {\bf Append}$(Ben_{\tt M})$, {\bf Append}$(Ann_{\tt F})$. \hfill {\scriptsize \tt//we add the current characters}
\STATE $Ben_{\tt P}=Ben_{\tt M}$, $Ben_{\tt M}={\bf read.in()}$.
\IF {$Ann_{\tt F}\neq (\textbf 2,\textbf d)$}	
	\STATE  $count=count{+}{+}$. \hfill {\scriptsize \tt//increment counter only if Ann's favourite is not $(\textbf 2,\textbf d)$}
\ENDIF
\ENDWHILE
\end{algorithmic}
\end{algorithm}

Let us first explain the above algorithm. Ann starts by choosing $(\textbf 0,\textbf a)$ as her favourite, and appends it. If the letter that Benn choses as first component is different from $\textbf 0$, then Ann only updates its second component and we append both letters. Otherwise, Ann will also update the first component to $\textbf 1$ (line 5). Both letters are appended. Obviously, the next update of the second component for Ann must consider the third position of the Thue-Morse word, thus we set in line 9 our counter to 3.

Now let us look at the big loop of the game, and consider all situations. 

If Ann has $(\textbf 2,\textbf d)$ as favourite, then the lines 11-12 say that the next choice that she has to make for the first component must be different, while the second component represents the next considered position of the Thue-Morse word. Since the current favourite of Ann was $(\textbf 2,\textbf d)$ it implies that the previous favourite was matched by Ben in the first component. Therefore, the new favourite will be one whose first component is different from the one of the current choice of Ben, and, moreover, different from the first component of his previous choice, whenever his current choice is $(\textbf 2,\textbf d)$.

If Ann has a favourite different from $(\textbf 2,\textbf d)$ and this is matched in the first component by the current choice of Ben, then Ann choses her current favourite such that its first component is different from the first components of both the current and previous choices of Ben (lines 13-14). Observe that in line 14, we know for sure that $0\leq 3-Ben_{\tt P}-Ben_{\tt M} < 3$ since the first component of the choice of Ben matches the one of Ann's favourite, which in turn, will definitely not match the first component of the previous choice of Ben.

If the first component of Ann's favourite is different from $\textbf 2$ and it is not match by Ben, then we only update its second component to the current considered position in the Thue-Morse word.

The loop ends by appending the current letters, updating the previous choice of Ben and reading the new one. The counter traversing the Thue-Morse word is updated only when the current favourite of Ann is different from $(\textbf 2,\textbf d)$.

We are now ready to state our result.
\begin{theorem}
There exists a strategy with finite description for Ann that allows her to win the non-repetitive game of any length on $7$ letters.
\end{theorem}
\begin{proof}
Assume that following this strategy there exist two consecutive repetitions of a word with length greater than $1$. Consider the following claim.

\medskip\noindent{\bf Claim.}\qquad
Any square-free word remains square-free after insertions of a new letter, as long as no unary squares are created.\\
{\bf Proof:}\qquad
The result is trivial since if a square is created, this would remain a square after deleting all occurrences of a letter from both sides.
\hfill\qed\medskip

Since Ann does not have the letter $(\textbf 2,\textbf d)$ as a favourite for two consecutive rounds, the previous claim establishes that if such a word exists, then it cannot have even length. Let us denote such a word by $x_1x_2\cdots x_n$ for some odd integer $n>2$. Moreover, let us denote by $i$ a position where such a repetition occurs in the word obtained while playing the game, namely $w$. Hence
$$w_{i+1}w_{i+2}\cdots w_{i+2n}=x_{1}x_{2}\cdots x_{n}x_{1}x_{2}\cdots x_{n}$$
Furthermore, observe that for every position $k$ with $1\leq k \leq n$ we have that $x_k[1]\in\{\textbf 0,\textbf 1,\textbf 2\}$ and $x_k[2]\in\{\textbf a,\textbf b,\textbf c,\textbf d\}$.
 
Assume that Ben starts. Then, following our strategy we have that 
\begin{equation*}
\begin{split}
w_{i+2}[1]\neq w_{i+1}[1], w_{i+4}[1]\neq w_{i+3}[1], \ldots,  w_{i+n-1}[1]\neq w_{i+n-2}[1],\\
w_{i+n+1}[1]\neq w_{i+n}[1], w_{i+n+3}[1]\neq w_{i+n+2}[1], \ldots,  w_{i+2n}[1]\neq w_{i+2n-2}[1],\\
\end{split}
\end{equation*}

Since $x_k=w_{i+k}=w_{i+k+n}$, we conclude that $x_k[1] \neq x_{k+1}[1]$ for every integer $k$ with $1\leq k < n$ and $x_1[1]\neq x_n[1]$. This implies that Ann was not constrained in fact to ever change her first letter, and therefore, throughout the factor she had only one favourite. Furthermore, Ben has used only letters different from the one that Ann used. This easily leads to a contradiction as every $x_k[1]$ is once chosen by Ann, and once by Benn.

Assume now that Ann puts the letter on position $w_i$. We once more conclude that $x_k[1] \neq x_{k+1}[1]$ for every integer $k$ with $1\leq k < n$, but in this case $x_1[1]$ might be equal to $x_n[1]$. In particular, the letter on position $w_{i+n+1}$ is chosen by Ben, and this has to match the one on position $w_i$. Therefore, Ann will change once, and only once, in the game her choice of letter. It easily follows that in this case either the length of the word that is squared is at most 3, or the choice of letter for Ben that preceded the enforced change of favourite of Ann is $(\textbf 2,\textbf d)$. Indeed if $n>3$ and $Ben_{\tt P}\neq (\textbf 2,\textbf d)$, then there is going to be another change of letters for Ann, which is a contradiction with the uniqueness of the change.

It is not difficult to check that there is no word of length $3$ whose square appears in our construction. For the second case, observe that since the last letter that Benn chose in the first half of the square must be $(\textbf 2,\textbf d)$, namely $x_{n-1}=(\textbf 2,\textbf d)$, we conclude that the last letter that Ann choses in the second half of the square, must be the same. However, in order for this to happen there must exist a second change of favourite for Ann, which is again a contradiction with the uniqueness of the choice.
\qed\end{proof}

\bibliographystyle{abbrv}
\bibliography{ref-Thue_games}
\end{document}